\documentclass[10pt,a4paper]{article}
\usepackage[utf8]{inputenc}
\usepackage[T1]{fontenc}
\input{glyphtounicode}
\usepackage{a4wide}
\usepackage{amsmath,amssymb,amsthm}
\IfFileExists{sfrm1000.pfb}{}{\InputIfFileExists{clt-fonts.tex}{}{}}
\IfFileExists{bbold.sty}{\usepackage{bbold}}{%
  \pdfmapline{+cltbb10 GXKDDN+BBOLD10 <cltbb10.pfb}
  \pdfmapline{+cltbb7 ZFJGRA+BBOLD7 <cltbb7.pfb}
  \DeclareFontFamily{U}{cltbb}{}
  \DeclareFontShape{U}{cltbb}{m}{n}{<5-8>cltbb7 <8->cltbb10}{}
  \DeclareMathAlphabet{\mathbb}{U}{cltbb}{m}{n}
}
\usepackage[colorlinks=true,linkcolor=red,citecolor=green,urlcolor=magenta]{hyperref}
\newcommand{\R}{\mathbb{R}}
\newcommand{\E}{\mathbb{E}}
\newcommand{\Prob}{\mathbb{P}}
\newcommand{\Q}{\mathbb{Q}}
\newcommand{\F}{\mathcal{F}}
\newcommand{\Id}{I_d}
\newcommand{\ind}{\mathbf{1}}
\newcommand{\Tr}{\mathrm{Tr}}
\DeclareMathOperator{\Var}{Var}
\newcommand{\laweq}{\overset{\text{\normalsize law}}{=}}
\newtheorem{theorem}{Theorem}
\newtheorem{lemma}[theorem]{Lemma}
\newtheorem{proposition}[theorem]{Proposition}
\newtheorem{corollary}[theorem]{Corollary}
\newtheorem{remark}[theorem]{Remark}
\title{Improved Bounds for the Entropic Central Limit Theorem\\
in the Log-Concave Case.}
\author{Tom Maitre}
\date{October 9, 2026}

\begin{document}
\maketitle
\begin{abstract}
Let $\mu$ be a log-concave isotropic probability measure on $\R^d$, and let $\mu^n$ denote the law of
the normalized sum of $n$ independent random vectors with law $\mu$. We prove that
\[
 H(\mu^n|\gamma)\leq \frac{C}{n}
 \min\left\{1+\log(ed)^2H(\mu|\gamma)\ ,\ d\right\}.
\]
Compared to the previous best known bound, our estimate replaces the polynomial dependence on the dimension by a polylogarithmic one. The proof relies on a new combination of
the De Bruijn identity with Eldan's stochastic localization process.
\end{abstract}

\section{Introduction and preliminaries}\label{sec:introduction}
Let $\mu$ be a probability measure on $\R^d$ with mean zero and identity covariance matrix and let
$(X_i)_{i\geq 1}$ be i.i.d.\ random vectors with law $\mu$. The Central Limit Theorem asserts that
\[
 \frac{1}{\sqrt{n}}\sum_{i=1}^n X_i
\]
converges in distribution to the standard Gaussian measure $\gamma$ on $\R^d$. In this paper, we study
convergence in relative entropy. For $\nu$ be a probability measure on $\R^d$. If $\mu$ is absolutely continuous
with respect to $\nu$, we define the relative entropy by
\[
 H(\mu|\nu)=\int_{\R^d}\log\left(\frac{d\mu}{d\nu}\right)d\mu.
\]
In \cite{Bar86}, Barron proved that if $H(\mu\mid\gamma)<\infty$, and letting $\mu^n$ be the law of
$\frac{1}{\sqrt n}\sum_{i=1}^n X_i$, then
it satisfies the entropic Central Limit Theorem,
\[
 \lim_{n\to\infty} H(\mu^n\mid\gamma)=0.
\]
We are interested in quantitative versions of this result. In dimension $d=1$, Ball, Barthe and
Naor \cite{Bal02} established entropy jump inequalities for measures satisfying a Poincaré inequality
and having finite entropy. Their estimates combined with the monotonicity of the relative entropy
(see \cite{Art04}) imply that, for every $n$,
\[
 H(\mu^n\mid\gamma)\leq \frac{2^\epsilon}{n^\varepsilon}H(\mu\mid\gamma),
\]
where $\varepsilon=-\log_2\left(1-\frac{1}{2+2C_P(\mu)}\right)$ and $C_P(\mu)$ denotes the Poincaré constant. Ball and Nguyen
\cite{Bal12} later extended this result to arbitrary dimensions, obtaining an analogous estimate but
only in the log-concave case.

In \cite{Cou19}, Courtade, Fathi and Pananjady established
\[
 H(\mu^n|\gamma)\leq\frac{dC_P(\mu)}{n}
 \log\left(1+\frac{I(\mu|\gamma)n}{d(C_P(\mu)-1)}\right)
\]
where $I(\mu|\gamma)$ is the relative Fisher information with respect to the standard Gaussian measure
(see below for the definition). This improves upon \cite{Bal12} in the sense that it does not require
log-concavity, and it has nearly optimal dependence on $n$. On the other hand, it requires finite
Fisher information, which is quite restrictive.

Eldan, Mikulincer, and Zhai \cite{Eld20} obtained a quantitative bound on the rate of convergence in
the log-concave and isotropic case :
\[
 H(\mu^n|\gamma)\leq C\frac{d^{10}}{n}\left(1+H(\mu|\gamma)\right),
\]
where $C$ is a universal constant. Their estimate improves upon result of Courtade and his coauthors by removing the extra logarithmic factor $\log n$, thereby obtaining the optimal $1/n$ rate,
and replacing the dependence on the Fisher information by a dependence on the relative entropy
yields a smaller right-hand side, so this is a stronger estimate in view of the Gaussian logarithmic
Sobolev inequality. On the other hand, the dependence on the dimension is much worse. Kannan,
Lovász and Simonovits proved in \cite{Kan95} that, if $\mu$ is the uniform probability measure on an
isotropic convex body in $\R^d$, then $C_P(\mu)=O(d)$. They conjectured that this estimate could be
replaced by the dimension-free bound $C_P(\mu)=O(1)$; in its usual formulation, this is conjectured
for every isotropic log-concave probability measure. The KLS conjecture has recently been proved
by Bizeul, Klartag and Lehec \cite{Biz26b} and \cite{Son26}. Consequently, the bound obtain by Courtad and al retains a
linear dependence on $d$ in its prefactor, which remains substantially better than the earlier $d^{10}$
dependence.

More recently, \cite{Den26} established the bound
\[
 H(\mu^n\mid\gamma)\leq\frac{2d\,C_P(\mu)}{n}H(\mu|\gamma),
\]
for all sufficiently large $n$ and every isotropic probability measure $\mu$ with finite Poincaré constant.
This significantly strengthens the result of Eldan and his coauthors, in two ways. On the one hand,
it does not require log-concavity. On the other hand, when the measure is log-concave, it improves
upon the dependence on the dimension.

Actually, in a sense the result dates back to Johnson
and Barron \cite{Joh04} in 2004, although the latter only considered the one-dimensional case. However, this extension to arbitrary dimensions is rather straightforward. It seems that Johnson and
Barron's result was overlooked by subsequent authors.

Our main result is the following theorem:
\begin{theorem}\label{thm:main}
Let $\mu$ be an isotropic log-concave probability measure on $\R^d$, and let $X$ be a random
vector with distribution $\mu$. Let $X_1,\ldots,X_n$ be i.i.d.\ copies of $X$ and denote by $\mu^n$ the distribution
of $\frac{1}{\sqrt n}\sum_{i=1}^n X_i$. Then
\[
 H(\mu^n|\gamma)\lesssim\frac{1}{n}
 \min\left\{1+\log^2(d)H(\mu|\gamma)\ ,\ d\right\}.
\]
\end{theorem}

The resolution of Bourgain's slicing problem implies that \(H(\mu\mid\gamma)\lesssim d\). Thus, when the initial entropy is of order \(d\), the second bound yields the rate \(d/n\), which is optimal in general. The first bound provides a refinement for measures with smaller initial entropy, improving this rate whenever
\(
H(\mu\mid\gamma)=o\!\left(\frac{d}{\log^2 d}\right).
\)
An ideal strengthening would be to remove both the logarithmic factor and the additive term, yielding
\[
H(\mu^n\mid\gamma)\leq \frac{C}{n}\,H(\mu\mid\gamma)
\]
for a universal constant \(C>0\). There is no apparent obstruction to such an estimate: it would recover the optimal order \(d/n\) when the initial entropy is comparable to \(d\), while providing an improvement throughout the regime \(H(\mu\mid\gamma)=o(d)\).
Compared with \cite{Den26}, this yields
a genuine gain since the slicing theorem (\cite{Kla24c}) implies $H(\mu|\gamma)\lesssim d$, so the improvement can
be of order $d$. The recent resolution of the KLS conjecture by Bizeul, Klartag and Lehec \cite{Biz26b}
makes the Poincaré constant universally bounded. Nevertheless, the estimate of \cite{Den26} retains
a linear dependence on $d$ in its prefactor, whereas our first bound involves only a polylogarithmic
factor. Our two estimates are therefore complementary: one improves the dimensional dependence,
while the other removes an entropy factor that may itself be of order $d$. Our approach, however,
relies crucially on log-concavity, and we do not currently know how to obtain comparable bounds
without this assumption.

The same strategy yields an improvement of the result of \cite{Eld20} in the isotropic, $\kappa$-uniformly log-concave setting, giving the following dimension-free bound. Note that the condition $\kappa\leq1$ follows
from the Brascamp--Lieb inequality and the isotropy of $\mu$, and therefore imposes no additional
restriction in the following theorem.
\begin{theorem}\label{thm:uniform}
Let $\kappa\in]0,1]$, and let $\mu$ be an isotropic $\kappa$-uniformly log-concave probability measure
on $\R^d$. With the notation of Theorem \ref{thm:main}, there exists a universal constant $C>0$ such that, for
every $n\geq1$,
\[
 H(\mu^n|\gamma)\leq\frac{C}{\kappa^2 n}H(\mu|\gamma).
\]
\end{theorem}
The isotropy assumption does not cause any loss of generality. Indeed, let $\mu$ be a centered log-concave probability measure with positive definite covariance matrix $\Sigma$, and let $\widetilde\mu=(\Sigma^{-1/2})_\#\mu$.
Then $\widetilde\mu$ is log-concave and isotropic. Moreover, if $\gamma_\Sigma$ denotes the centered Gaussian measure with
covariance matrix $\Sigma$, the invariance of the relative entropy under invertible linear transformations
gives
\[
 H(\widetilde\mu|\gamma)=H(\mu|\gamma_\Sigma)
 \qquad\text{and}\qquad
 H(\widetilde\mu^{\,n}|\gamma)=H(\mu^n|\gamma_\Sigma).
\]
The paper is organized as follows. In Section 1 we recall several results on stochastic localization.
In Section 2 we prove a variant of the De Bruijn formula using elementary tools from stochastic
localization. In Section 3 we establish a variational representation of the entropy using stochastic
calculus. Finally, in Section 4 we recall additional properties of log-concave measures and combine
the previous results to prove Theorem \ref{thm:main} and Theorem \ref{thm:uniform}.

\section{Stochastic localization}\label{sec:localization}
Let $\mu$ be a probability measure on $\R^d$ and $t>0$. We denote by $\rho$ the density of $\mu$ and by $\gamma_t$ the
centered Gaussian measure with covariance matrix $t\Id$. Moreover, we write $\rho_t$ for the density of
the convolution $\mu*\gamma_t$. We define the Laplace transform of $\mu$ by, for all $\theta\in\R^d$ and $t>0$,
\begin{equation}\label{eq:laplace}
 \Lambda(t,\theta)=\log\int_{\R^d}
 \exp\left(\langle x,\theta\rangle-t|x|^2/2\right)\mu(dx).
\end{equation}
For simplicity, we assume that $e^{\langle x,\theta\rangle}$ is integrable with respect to $\mu$, for all $\theta\in\R^n$. In this case,
$\Lambda(t,\theta)$ is $C^\infty$ on $[0,\infty)\times\R^d$. We also introduce the measure $(\mu_{t,\theta})$ defined by
\begin{equation}\label{eq:tilt}
 d\mu_{t,\theta}(x)=
 \exp\left(-\Lambda(t,\theta)+\langle x,\theta\rangle-t|x|^2/2\right)\mu(dx).
\end{equation}
We denote by $a(t,\theta)$ the barycenter of the measure $\mu_{t,\theta}$ and by $A(t,\theta)\in\R^{d\times d}$ the covariance
matrix associated to $\mu_{t,\theta}$, defined as
\[
 A(t,\theta)=\int_{\R^d}x^{\otimes2}\mu_{t,\theta}(dx)-a(t,\theta)^{\otimes2}.
\]
Moreover, we introduce the process $(\theta_t)$ given by
\begin{equation}\label{eq:observation}
 \theta_t=tX+B_t
\end{equation}
where $(B_t)$ is a Brownian motion, and $X$ has distribution $\mu$, and is independent of $(B_t)$. We
abbreviate $a_t=\nabla\Lambda(t,\theta_t)$ the barycenter of $\mu_{t,\theta_t}$ and $A_t=\nabla^2\Lambda(t,\theta_t)$ the covariance matrix
process of $\mu_{t,\theta_t}$. We denote by $(\F_t)$ the natural filtration generated by $(\theta_t)$.

\begin{lemma}\label{lem:conditional}
Let $\phi:\R^d\to\R$ be a bounded continuous function. We have
\[
 \int_{\R^d}\phi\,d\mu_{t,\theta_t}
 =\E[\phi(X)|\F_t]=\E[\phi(X)|\theta_t].
\]
In particular, the process $\left(\int_{\R^d}\phi\,d\mu_{t,\theta_t}\right)_{t\geq0}$ is a martingale with respect to the filtration $(\F_t)_t$ generated by $(\theta_s)_{s\leq t}$
\end{lemma}
\begin{proof}
Let $t>0$. As $\theta_s/s=X+B_s/s=X+\widetilde B_u$ where $u=\frac1s$ and $\widetilde B_u$ is a Brownian motion
independent of $X$, we have
\[
 \E[f(X)|\F_t]=\E[f(X)|(\theta_s)_{s\leq t}]
 =\E[f(X)\mid(\theta_s/s)_{s\leq t}]
 =\E[f(X)|(X+\widetilde B_u)_{u\geq1/t}].
\]
Observe that
\[
 \sigma(X+\widetilde B_u,u\geq1/t)=
 \sigma(X+\widetilde B_{1/t})\vee
 \sigma(\widetilde B_u-\widetilde B_{1/t},u\geq1/t)
\]
and that by Markov property of Brownian motion, the process $(\widetilde B_u-\widetilde B_{1/t})_{u\geq1/t}$ is independent
of $(X,\widetilde B_{1/t})$. Therefore,
\begin{equation}\label{eq:time-inversion}
 \E[f(X)|(X+\widetilde B_u)_{u\geq1/t}]
 =\E[f(X)|X+\widetilde B_{1/t}].
\end{equation}
Let $g$ be a measurable function, and assume that $\mu$ has a density $\rho$. As the heat semigroup $(P_t)$
is self-adjoint in $L^2(dx)$, we obtain
\begin{align*}
 \E[f(X)g(X+B_t)]&=\int_{\R^d}f(x)P_t(g)(x)\rho(x)\,dx\\
 &=\int_{\R^d}g(x)P_t(f\rho)(x)dx\\
 &=\E[Q_t f(X+B_t)g(X+B_t)]
\end{align*}
where $Q_t(f)=\frac{P_t(f\rho)}{P_t(\rho)}$. By definition of the conditional expectation, we obtain, for all $t>0$,
\begin{equation}\label{eq:heat-conditional}
 \E[f(X)|X+B_t]=Q_t f(X+B_t).
\end{equation}
Applying \eqref{eq:time-inversion} and \eqref{eq:heat-conditional} with the Brownian motion $(\widetilde B_{1/t})$ yields
\[
 \E[f(X)|\F_t]=\E[f(X)\mid X+\widetilde B_{1/t}]
 =Q_{1/t}f(X+\widetilde B_{1/t}).
\]
It remains to compute the latter term. By \eqref{eq:tilt}, denoting by $\rho$ the density of $\mu$, we have
\[
 \int_{\R^d}f\,d\mu_{t,\theta_t}
 =\frac{\int(f\rho)(x)e^{\langle x,\theta_t\rangle-t|x|^2/2}dx}
 {\int\rho(x)e^{\langle x,\theta_t\rangle-t|x|^2/2}dx}
 =\frac{P_{1/t}(f\rho)(\theta_t/t)}{P_{1/t}\rho(\theta_t/t)}
\]
which implies
\[
 Q_{1/t}f(X+\widetilde B_{1/t})=\int_{\R^d}f\,d\mu_{t,\theta_t}
\]
since $\theta_t/t=X+\widetilde B_{1/t}$.
\end{proof}
The next lemma collects two basic moment computations for the process $(\theta_t)$ and its conditional
expectation $a_t=\E[X|\theta_t]$, which will play a key role in the analysis below.
\begin{lemma}\label{lem:moments}
With the same notation as above, we have
\[
 \E\langle a_t,\theta_t\rangle=t\E|X|^2
\]
and
\[
 \E|\theta_t|^2=t^2\E|X|^2+td.
\]
\end{lemma}

\begin{proof}
As $a_t=\E[X|\theta_t]$, we have, for all $t\geq0$,
\[
 \E\langle a_t,\theta_t\rangle
 =\E\langle\theta_t,\E[X|\theta_t]\rangle
 =\E\langle X,\theta_t\rangle=t\E|X|^2
\]
by independence of $X$ and the Brownian motion $(B_t)$ and \eqref{eq:observation}. Moreover, \eqref{eq:observation} also yields
\[
 \E|\theta_t|^2=t^2\E|X|^2+td.\qedhere
\]
\end{proof}
The next lemma ensures that the process $(a_t)$ converges in law to $\mu$ as $t\to\infty$ :
\begin{lemma}\label{lem:limit}
With the same notation as above,
\[
 \E\,|a_t-X|^2\leq\frac{d}{t}.
\]
In particular, $a_t\xrightarrow[t\to\infty]{}X$ in $L^2$, and $a_t\xrightarrow[t\to\infty]{}\mu$ in law.
\end{lemma}
\begin{proof}
By Lemma \ref{lem:conditional}, we have the representation
\[
 a_t=\E[X|\F_t].
\]
Moreover, the process $(\theta_t/t)$ is $(\F_t)$-adapted and has finite second moments. By the characterization of conditional expectation as an orthogonal projection onto $L^2(\F_t)$,
\[
 \E\,|a_t-X|^2
 \leq\E\left|\frac{\theta_t}{t}-X\right|^2
 =\E\left|\frac{B_t}{t}\right|^2=\frac{d}{t}
\]
as desired.
\end{proof}
It is possible to relate the process $\E A_t$ to the random vector $X$ with law $\mu$. Indeed, we have
\[
 A_t=\E[X\otimes X|\theta_t]-\E[X|\theta_t]\otimes\E[X|\theta_t]
 =\E[X\otimes X|\theta_t]-a_t\otimes a_t.
\]
Taking expectations, we get
\begin{equation}\label{eq:mean-covariance}
 \E A_t=\E[X\otimes X]-\E[a_t\otimes a_t].
\end{equation}

\section{A variant of the De Bruijn formula}\label{sec:debruijn}
The Shannon entropy of $\mu$ is defined as
\[
 S(\mu)=-\E\ln(\phi(X))
\]
where $\phi=\frac{d\mu}{dx}$.The relative entropy is related to the differential entropy through the identity
through the identity
\begin{equation}\label{eq:shannon-relative}
 S(\gamma)-S(\mu)+\frac12(\E|X|^2-d)=H(\mu|\gamma).
\end{equation}
It is well-known that the derivative of the entropy along the heat semigroup, which does appear
in the Fisher information, defined as
\[
 I(\mu)=\int_{\R^d}|\nabla\log f|^2f\,dx,
\]
where $f$ is the density of $\mu$ with respect to Lebesgue measure. More precisely, in \cite{Car91}, it is
shown that if $\mu$ has a finite second moment and finite entropy, De Bruijn relation holds
\begin{equation}\label{eq:debruijn}
 S(\gamma)-S(\mu)=\frac12\int_0^\infty
 \left[I(\mu P_s)-\frac{d}{1+s}\right]ds
\end{equation}
where $(P_t)$ is the heat semigroup. We reformulate \eqref{eq:debruijn} within the framework of the stochastic
localization process associated with $\mu$. We use the notation $(a_t)$, $(A_t)$ and $(\theta_t)$ from the previous
section.

\begin{lemma}\label{lem:entropy-barycenter}
Let $\mu$ be a probability measure on $\R^d$ with finite second moments and finite entropy.
Then,
\begin{equation}\label{eq:entropy-barycenter}
 H(\mu|\gamma)=\frac12\int_0^\infty
 \E\left|a_t-\frac{\theta_t}{1+t}\right|^2dt
\end{equation}
where $\lvert\cdot\rvert$ denotes the Euclidean norm.
\end{lemma}
\begin{proof}
Let $(P_s)$ denote the heat semigroup, $\rho$ the density of $\mu$ and $\rho_s$ the density of $\mu P_s=\mu*\gamma_s$.
Let $X$ be a random vector with distribution $\mu$ and $(B_s)_{s\geq0}$ be a Brownian motion independent of
$X$. Then, the random vector $X_s=X+B_s$ has distribution $\mu*\gamma_s$. For all $s>0$,
\[
 P_s\rho(x)=(2\pi s)^{-d/2}\int_{\R^d}\rho(y)\cdot
 \exp\left(-\frac{|x-y|^2}{2s}\right)dy.
\]
Recall the definition \eqref{eq:laplace} of $\Lambda$ and notice that
\[
 \log P_s\rho(x)=-\frac d2\log(2\pi s)-\frac{|x|^2}{2s}
 +\Lambda\left(\frac1s,\frac xs\right).
\]
Therefore\vspace{\baselineskip}
\begin{equation}\label{eq:fisher-localization}
 I(\mu P_s)=\E|\nabla\log P_s\rho(X_s)|^2
 =\E\left|\frac1s\nabla\Lambda\left(\frac1s,\frac{X_s}{s}\right)
 -\frac{X_s}{s}\right|^2.
\end{equation}
The process $s^{-1}X_s=s^{-1}X+s^{-1}B_s$ has the same law as $tX+B_t=\theta_t$ with $t=s^{-1}$. Recall that
$a_t=a(t,\theta_t)=\nabla\Lambda(t,\theta_t)$. Then by \eqref{eq:debruijn}, \eqref{eq:fisher-localization} and making the change of variables $t=s^{-1}$, we get
\begin{align}
 S(\gamma)-S(\mu)
 &=\frac12\int_0^\infty\left[
 \E\left|s^{-1}a(s^{-1},s^{-1}X_s)-s^{-1}X_s\right|^2
 -\frac{d}{1+s}\right]ds\notag\\
 &=\frac12\int_0^\infty\left[\E\,|ta_t-\theta_t|^2
 -\frac{d}{1+t^{-1}}\right]\frac{dt}{t^2}\notag\\
 &=\frac12\int_0^\infty\left[
 \E\left|a_t-\frac{\theta_t}{t}\right|^2
 -\frac{d}{t(1+t)}\right]dt.\label{eq:debruijn-time-change}
\end{align}
We observe that, by Lemma \ref{lem:moments},
\begin{equation}\label{eq:moment-identity}
 \E\left|a_t-\frac{\theta_t}{t}\right|^2-\frac{d}{t(1+t)}
 =\E|a_t|^2-\E|X|^2+\frac{d}{1+t}
\end{equation}
and, by Lemma \ref{lem:moments} and \eqref{eq:moment-identity},
\begin{align*}
 \E\left|a_t-\frac{\theta_t}{1+t}\right|^2
 &=\E|a_t|^2+
 \left(\frac{1}{(1+t)^2}-1\right)\E|X|^2
 +\frac{t}{(1+t)^2}d\\
 &=\left(\E\left|a_t-\frac{\theta_t}{t}\right|^2
 -\frac{d}{t(1+t)}\right)
 +\frac{\E|X|^2-d}{(1+t)^2}.
\end{align*}
Integrating this relation over $[0,\infty)$ and using \eqref{eq:debruijn-time-change} yields
\[
 S(\gamma)-S(\mu)=\frac12\int_0^\infty
 \E\left|a_t-\frac{\theta_t}{1+t}\right|^2dt+\frac12(d-\E|X|^2)
\]
which is the result by \eqref{eq:shannon-relative}.
\end{proof}

\section{A variational formula for the entropy}\label{sec:variational}
In this section, we present a theorem originating from the work of \cite{Eld20}, formulated here with
a slightly different time parametrization adapted to our setting. Its proof relies on a variational
representation of the relative entropy established by \cite{Leh13} using stochastic calculus techniques.
\begin{proposition}\label{prop:variational}
Let $\mu$ be a probability measure. Let $(B_t)$ be a Brownian motion and let $(\Gamma_t)$ be an
adapted process taking values in the space of symmetric matrices. Let $(m_t)_{t\geq0}$ be a deterministic
path with value in the space of symmetric matrices and satisfying
\[
 \int_0^\infty m_t^2\,dt=\Id.
\]
We assume that
\[
 \int_0^\infty\Gamma_t\,dB_t\laweq\mu.
\]
Then
\[
 H(\mu|\gamma)\leq\frac12\E\int_0^\infty\Tr\left[
 \left(\int_t^\infty m_s^2\,ds\right)^{-1}(m_t-\Gamma_t)^2\right]dt.
\]
\end{proposition}
\begin{proof}
We first prove that if $(u_t)$ is an adapted process with values in $\R^d$ such that
\begin{equation}\label{eq:drift-embedding}
 \int_0^\infty m_t(dB_t+u_t\,dt)\laweq\mu
\end{equation}
then,
\begin{equation}\label{eq:drift-cost}
 H(\mu|\gamma)\leq\frac12\E\left[\int_0^\infty|u_t|^2dt\right].
\end{equation}
To see this, let $X=\int_0^\infty m_t(dB_t+u_t\,dt)$. By \eqref{eq:drift-embedding}, the random vector $X$ has law $\mu$. Next, assume
first that $\int_0^\infty|u_s|^2ds\in L^\infty$ and consider the change of measure
\[
 d\Q=\exp\left(-\int_0^\infty u_s\,dB_s-\frac12\int_0^\infty|u_s|^2ds\right)d\Prob.
\]
By Novikov's criterion, the above stochastic exponential defines a martingale and we can apply
Girsanov's theorem. Under $\Q$, the process $B_t+\int_0^t u_s\,ds$ becomes a Brownian motion. Since
$\int m_t^2dt=1$, under $\Q$, $X$ is distribued according to $\gamma$. Therefore, applying the data processing
inequality yields
\[
 H(\mu|\gamma)=H(X_\#\Prob|X_\#\Q)\leq H(\Prob|\Q).
\]
By the definition of the relative entropy and using the fact that $\int_0^\infty u_t\cdot dB_t$ has mean 0, it follows
that
\[
 H(\Prob|\Q)=\E\ln\frac{d\Prob}{d\Q}
 =\frac12\E\int_0^\infty|u_t|^2dt.
\]
For the general case, when $\int_0^\infty|u_t|^2dt$ is not necessarily bounded, follows from the lower semicontinuity of the map $\mu\to H(\mu|\nu)$. See the proof of Proposition 1 in \cite{Leh13} for details.

\medskip
We are now in a position to prove the proposition. Let $(\Gamma_s)$ be a symmetric matrix process such
that $\int_0^\infty\Gamma_s\,dB_s\sim\mu$. Let $X=\int_0^\infty\Gamma_s\,dB_s$. We need to write
\begin{equation}\label{eq:desired-drift}
 X=\int_0^\infty m_t(dB_t+u_t\,dt),
\end{equation}
for some adapted process $(u_t)$. We take $u$ of the form
\begin{equation}\label{eq:drift-form}
 u_t=m_t\int_0^t H_s\,dB_s
\end{equation}
where $(H_t)$ is an adapted matrix-valued process. Using Fubini's theorem,
\begin{align*}
 \int_0^\infty m_t(dB_t+u_t\,dt)
 &=\int_0^\infty m_t(dB_t+m_t\left(\int_0^t H_s\,dB_s\right)dt)\\
 &=\int_0^\infty m_t\,dB_t+
 \int_0^\infty\left(\int_s^\infty m_t^2\,dt\right)H_s\,dB_s\\
 &=\int_0^\infty\left(m_s+\left(\int_s^\infty m_t^2\,dt\right)H_s\right)dB_s.
\end{align*}
Identifying the integrands, we obtain, for all $s>0$,
\[
 m_s+\left(\int_s^\infty m_t^2\,du\right)H_s=\Gamma_s.
\]
Therefore
\begin{equation}\label{eq:H-process}
 H_t=\left(\int_t^\infty m_s^2\,ds\right)^{-1}(\Gamma_t-m_t).
\end{equation}
If $(u_t)$ is defined by \eqref{eq:drift-form} and \eqref{eq:H-process}, then
\[
 X=\int_0^\infty\Gamma_t\,dB_t=\int_0^\infty m_t(dB_t+u_t\,dt).
\]
On the other hand, by Itô's isometry,
\begin{equation}\label{eq:ito-drift}
 \E|u_t|^2=\Tr\left(m_t^2\int_0^t\E[H_sH_s^T]ds\right).
\end{equation}
Thus, by \eqref{eq:drift-cost} and \eqref{eq:ito-drift},
\begin{align*}
 H(\mu|\gamma)&\leq\frac12\int_0^\infty\E|u_t|^2dt\\
 &=\frac12\int_0^\infty\Tr\left(m_t^2\int_0^t\E[H_sH_s^T]ds\right)dt\\
 &=\frac12\Tr\left(\int_0^\infty\int_s^\infty m_t^2\E[H_sH_s^T]dsdt\right)\\
 &=\frac12\Tr\left(\int_0^\infty\E\left[
 \left(\int_s^\infty m_u^2\,du\right)H_sH_s^T\right]ds\right)\\
 &=\frac12\int_0^\infty\Tr\,\E\left[
 \left(\int_s^\infty m_u^2\,du\right)^{-1}(\Gamma_s-m_s)^2\right]ds.\qedhere
\end{align*}
\end{proof}

Let us consider the special case $m_t=\frac1{1+t}\Id$ and establish a connection with the stochastic localization. Again, we use the definitions of the stochastic localization : $(a_t)$, $(A_t)$, $(\theta_t)$, $(\mu_{t,\theta_t})$ defined
in the previous sections and $(\F_t)$ the natural filtration generated by $(\theta_t)$.
\begin{lemma}\label{lem:innovation}
There exists a $(\F_t)$ Brownian motion $(W_t)$ such that
\begin{equation}\label{eq:innovation}
 \theta_t=W_t+\int_0^t a(s,\theta_s)ds.
\end{equation}
Moreover, we have
\[
 a_t=a_0+\int_0^t A_s\,dW_s.
\]
\end{lemma}
See e.g.\ \cite{Kla21} or \cite{Kla25} for a proof of this result. The process $(\mu_t)=(\mu_{t,\theta_t})$ is called the
stochastic localization. It follows from the previous results that the relative entropy admits a
representation in terms of the stochastic localization process :

\begin{proposition}\label{prop:entropy-covariance}
Let $\mu$ be a probability measure on $\R^n$ with finite second moments, centered and
finite entropy. Then,
\begin{equation}\label{eq:covariance-embedding}
 \int_0^\infty A_t\,dW_t\laweq\mu
\end{equation}
and on the other hand
\begin{equation}\label{eq:entropy-covariance}
 H(\mu|\gamma)=\frac12\int_0^\infty(1+t)\,
 \E\Tr\left[\left(\frac{\Id}{1+t}-A_t\right)^2\right]dt.
\end{equation}
In other words, when $m_t=\frac1{1+t}\Id$, the process $(A_t)$ achieves equality in Proposition $\ref{prop:variational}$.
\end{proposition}
\begin{proof}
Firstly, when $\mu$ is centered we have $a_0=0$ and Lemma \ref{lem:innovation} yields
\[
 a_t=\int_0^t A_s\,dW_s.
\]
By Lemma \ref{lem:limit}, the process $(a_t)$ converges in $L^2$, and its limit has law $\mu$. Consequently
\[
 \int_0^\infty A_t\,dW_t\laweq\mu.
\]
We now turn to the second statement. By \eqref{eq:innovation}, the process $((1+t)a_t-\theta_t)$ satisfies $d[(1+t)a_t-\theta_t]=[(1+t)A_t-\Id]dW_t$. Then, by Itô's isometry and since $a_0=\E X=0$ and $\theta_0=0$,
\[
 \E|(1+t)a_t-\theta_t|^2
 =\int_0^t\E\Tr\left[((1+s)A_s-\Id)^2\right]ds.
\]
Plugging this into Lemma \ref{lem:entropy-barycenter},
\begin{align*}
 H(\mu|\gamma)
 &=\frac12\int_0^\infty\left(\int_0^t
 \E\Tr\left[((1+s)A_s-\Id)^2\right]ds\right)\frac{dt}{(1+t)^2}\\
 &=\frac12\int_0^\infty\E\Tr\left[((1+s)A_s-\Id)^2\right]
 \left(\int_s^\infty(1+t)^{-2}dt\right)ds\\
 &=\frac12\int_0^\infty\E\Tr\left[((1+s)A_s-\Id)^2\right]\frac{ds}{1+s}\qedhere
\end{align*}
\end{proof}
Next, we present a technical matrix lemma, which can be found in \cite{Eld20}. The lemma is due to
these authors, and the proof provided in the cited reference contains an error.
\begin{lemma}\label{lem:square-root}
Let A,B be symmetric matrices, with A positive definite and B positive semi-definite. We have the inequality
\[
 \Tr\left[(\sqrt A-\sqrt B)^2\right]\leq
 \Tr\left[(A-B)^2A^{-1}\right].
\]
\end{lemma}
\begin{proof}
We have
\[
 A-B=A^{1/2}\left(A^{1/2}-B^{1/2}\right)
 +\left(A^{1/2}-B^{1/2}\right)B^{1/2}.
\]
Let $M_1=A^{-1/2}(A-B)$, $M_2=A^{1/2}-B^{1/2}$ and
$M_3=A^{-1/2}(A^{1/2}-B^{1/2})B^{1/2}$. By the
previous equation, we have
\begin{equation}\label{eq:matrix-decomposition}
 M_1=M_2+M_3.
\end{equation}
Then, by \eqref{eq:matrix-decomposition} and using symmetry and the cyclicity of the trace, we obtain
\begin{align*}
 \Tr((A-B)^2A^{-1})&=\Tr(M_1M_1^*)\\
 &=\Tr(M_2M_2^*)+\Tr(M_3M_3^*)+\Tr(M_2M_3^*)+\Tr(M_3M_2^*)\\
 &=\Tr(M_2M_2^*)+\Tr(M_3M_3^*)+2\Tr(M_2M_3^*).
\end{align*}
Since $\Tr(M_3M_3^*)\geq0$ and
\begin{align*}
 \Tr(M_2M_3^*)&=\Tr\left(A^{-1/4}M_2B^{1/2}M_2A^{-1/4}\right)\\
 &=\Tr\left((A^{-1/4}M_2B^{1/4})(A^{-1/4}M_2B^{1/4})^*\right)\geq0,
\end{align*}
we have, by symmetry,
\[
 \Tr((A-B)^2A^{-1})\geq\Tr(M_2M_2^*)=\Tr((\sqrt A-\sqrt B)^2).\qedhere
\]
\end{proof}
We now turn to the central limit theorem. The results established in this section allow us to obtain
a quantitative control of the entropy convergence in the CLT. Up to this point, our argument
follows the approach of \cite{Eld20}.
\begin{corollary}\label{cor:entropy-clt}
Let $\mu$ be an isotropic probability measure on $\R^d$, and let $(A_t)$ denote the covariance
matrix process associated with the stochastic localization. Then
\begin{equation}\label{eq:entropy-clt}
 H(\mu^n|\gamma)\leq\frac1{2n}\int_0^\infty
 \|( \E A_t)^{-1}\|\,\E\Tr\left[
 (A_t^2-\E A_t^2)^2(\E A_t^2)^{-1}\right]dt.
\end{equation}
\end{corollary}
\begin{proof}
Let $X_1,\ldots,X_n$ be i.i.d.\ random vectors with common distribution $\mu$, and let $B_t^1,\ldots,B_t^n$
be independent standard Brownian motions, such that the pairs $(X_i,B^i)$ are independents, and
$B^i,X_i$ are independents, for all $i$. For each $i$, define
\[
 \theta_t^i=tX_i+B_t^i,
\]
and denote by $(\F_t^i)$ the natural filtration generated by $\theta_t^i$. Let $(W_t^i)$ be the Brownian motion
associated with the representation \eqref{eq:innovation} for the process $\theta_t^i$, and denote by $A_{i,t}$ the corresponding
covariance matrix process. In particular, by \eqref{eq:innovation}, we have the stochastic integral representation
\[
 X_i=\int_0^\infty A_{i,t}\,dW_t^i,\qquad i=1,\ldots,n.
\]
Thus
\[
 \frac1{\sqrt n}\int_0^\infty\sum_{i=1}^n A_{i,t}\,dW_t^i
 =\frac1{\sqrt n}\sum_{i=1}^n X_i\laweq\mu^n.
\]
Next, define the process $(\Gamma_t)=
\left(\left(\frac1n\sum_{i=1}^n A_{i,t}^2\right)^{1/2}\right)$.
For all $t\geq0$, we can then write
\begin{align*}
 \frac1{\sqrt n}\sum_{i=1}^n A_{i,t}\,dW_t^i
 &=\frac1{\sqrt n}\left(\sum_{i=1}^n A_{i,t}^2\right)^{1/2}
 \left(\sum_{i=1}^n A_{i,t}^2\right)^{-1/2}
 \sum_{i=1}^n A_{i,t}\,dW_t^i\\
 &=\frac1{\sqrt n}\left(\sum_{i=1}^n A_{i,t}^2\right)^{1/2}d\widetilde W_t.
\end{align*}
where $d\widetilde W_t=\sum_{i=1}^n
\left(\sum_{j=1}^n A_{j,t}^2\right)^{-1/2}A_{i,t}dW_t^i$ is a continuous martingale starting at 0. Its quadratic
covariation is
\[
 d[\widetilde W]_t=
 \left(\sum_{i=1}^n A_{i,t}^2\right)^{-1}
 \left(\sum_{i=1}^n A_{i,t}^2\right)dt=\Id\cdot dt.
\]
By Lévy's characterization, $W$ is a standard Brownian motion. Moreover, $(\Gamma_t)$ is an adapted
process such that
\[
 \int_0^\infty\Gamma_t\,d\tilde W_t\laweq\mu^n.
\]
Next, by \eqref{eq:mean-covariance}, $\frac{d}{dt}\E A_t=-\frac{d}{dt}\E[a_t\otimes a_t]$. Lemma \ref{lem:innovation} and Itô's isometry imply
\[
 \E[a_ta_t^T]
 =\E\left[\left(\int_0^t A_s\,dW_s\right)
 \left(\int_0^t A_s\,dW_s\right)^T\right]
 =\E\left[\int_0^t A_s^2ds\right].
\]
Then,
\begin{equation}\label{eq:mean-derivative}
 \frac{d}{dt}\E A_t=-\frac{d}{dt}\E[a_t\otimes a_t]=-\E A_t^2.
\end{equation}
Moreover, by \eqref{eq:mean-derivative} together with Lemma \ref{lem:limit}, the function $t\mapsto\E A_t$ converges to 0. Then, the process
$\left(\E A_t^2\right)_t$ satisfies
\[
 \int_t^\infty\E A_s^2\,ds
 =-\int_t^\infty\frac{d}{ds}\E A_s\,ds=\E A_t.
\]
Hence, $(m_t)=(\E A_t^2)^{1/2}$ satisfies
\[
 \int_0^\infty m_t^2dt=\E A_0=\Id
\]
by isotropy. Now, Proposition \ref{prop:variational} allows us to obtains the following bounds on relative entropy of
$\mu^n$ :
\vspace{10.06pt}
\begin{align}
 H(\mu^n|\gamma)&\leq\frac12\int_0^\infty
 \E\Tr\left[(\E A_t)^{-1}
 \left(\left(\frac1n\sum A_{i,t}^2\right)^{1/2}
 -(\E A_t^2)^{1/2}\right)^2\right]dt\notag\\
 &\leq\frac12\int_0^\infty
 \|(\E A_t)^{-1}\|\,\E\Tr\left[
 \left(\left(\frac1n\sum A_{i,t}^2\right)^{1/2}
 -(\E A_t^2)^{1/2}\right)^2\right]dt.\label{eq:entropy-before-square-root}
\end{align}
Applying Lemma \ref{lem:square-root} with $A=\E A_t^2$ and $B=\frac1n\sum_{i=1}^n A_{i,t}^2$ which are symmetric and positive
semidefinite, we obtain
\[
 \E\Tr\left[\left(\left(\frac1n\sum_{i=1}^n A_{i,t}^2\right)^{1/2}
 -(\E A_t^2)^{1/2}\right)^2\right]
 \leq\E\Tr\left[\left(\frac1n\sum_{i=1}^n A_{i,t}^2-\E A_t^2\right)^2
 (\E A_t^2)^{-1}\right].
\]
Since the matrices $(A_{i,t})_i$ are independent copies of $A_t$, we have
\begin{align*}
 \E\Tr\left[\left(\frac1n\sum_{i=1}^n
 \left(A_{i,t}^2-\E A_t^2\right)\right)^2
 (\E A_t^2)^{-1}\right]
 &=\frac1{n^2}\sum_{i=1}^n\E\Tr\left[
 \left(A_{i,t}^2-\E A_t^2\right)^2(\E A_t^2)^{-1}\right]\\
 &=\frac1n\Tr\left[\E\left(A_t^2-\E A_t^2\right)^2(\E A_t^2)^{-1}\right].
\end{align*}
Together with \eqref{eq:entropy-before-square-root}, the previous equation yields the desired result.\\
\end{proof}

\section{Bound on entropic convergence in the log-concave case}\label{sec:bounds}
We can improve the control of the covariance matrix for a log-concave probability measure in the
stochastic localization process. Indeed, the measure $\mu_t$ is $t$-uniformly log-concave, and therefore
its covariance matrix can be bounded as
\begin{equation}\label{eq:pointwise-covariance}
 A_t\preccurlyeq\frac1t\Id,\qquad\forall t\in]0,\infty[.
\end{equation}

Combining the covariance estimate of \cite{Kla24}[Theorem 61] with the results of \cite{Che26}, as explained in \cite{Biz26}, yields the following concentration inequality for the covariance process \((A_t)_{t\geq 0}\). There exists a universal constant \(c>0\) such that, for every isotropic log-concave probability measure \(\mu\) on \(\mathbb R^d\), with \(d\geq 2\),
\begin{equation}\label{eq:concentration}
\mathbb P\bigl(\|A_t\|_{\mathrm{op}}\geq 2\bigr)
\leq e^{-\frac{c}{t}},\end{equation} for all
$0<t\leq t_0:=\frac{c}{\log d}.$
We will also use the following uniform fourth-moment estimate for the covariance process \((A_t)_{t\geq 0}\), which follows from \cite{Gua25}[Corollary 2.4]; see also \cite{Kla25}[Section 5]. For every isotropic log-concave probability measure \(\mu\) on \(\mathbb R^d\),
\begin{equation}\label{eq:fourth-moment}
\mathbb E\operatorname{Tr}(A_t^4)\leq Cd,
\qquad \forall t\geq 0,\end{equation}
where \(C>0\) is a universal constant. 

To begin, we provide an upper bound on the norm of the inverse expected covariance:
\begin{lemma}\label{lem:inverse-mean}
Assume that $\mu$ is log-concave and isotropic. Then, for every $t\geq0$,
\[
 \|(\E A_t)^{-1}\|\leq1+C_P(\mu)t.
\]
\end{lemma}
\begin{proof}
Recall the variance decomposition
\begin{equation}\label{eq:variance-decomposition}
 \Var f(X)=\Var(\E[f(X)\mid Y])+\E[\Var(f(X)\mid Y)].
\end{equation}
In \cite[Theorem 1.1]{Kla24b}, it is shown that, for any probability measure $\mu$ and any square-integrable
function $f$,
\[
 \Var\left(\E\left[f(X)\mid X+\sqrt s\,G\right]\right)
 \le \frac{C_P(\mu)}{s+C_P(\mu)}\Var f(X),
\]
where $X$ has law $\mu$ and $G$ is a standard Gaussian random vector independent of $X$. Applying this
estimate with the time change $s=1/t$, we obtain
\begin{align*}
 \Var(\E[f(X)\mid\theta_t])
 &=\Var(\E[f(X)\mid\theta_t/t])\\
 &=\Var\left(\E\left[f(X)\,\middle|\,X+\sqrt{\frac1t}\,G\right]\right)\\
 &\le \frac{tC_P(\mu)}{1+tC_P(\mu)}\Var f(X),
\end{align*}
where $(\theta_t)$ is defined in \eqref{eq:observation}. Combining this with \eqref{eq:variance-decomposition} yields
\[
 \E[\Var(f(X)\mid\theta_t)]\ge \frac1{1+tC_P(\mu)}\Var f(X).
\]
Applying this inequality to linear functions gives, for all $t\geq0$,
\[
 \E A_t\ge \frac1{1+tC_P(\mu)}A_0
 =\frac1{1+tC_P(\mu)}\Id,
\]
by isotropy of $\mu$.
\end{proof}
Before proving our main theorem, we establish a technical lemma that bounds the quantity in
\eqref{eq:entropy-clt} :

\begin{lemma}\label{lem:matrix-fluctuation}
Let $M$ be a random symmetric matrix such that $0\preccurlyeq M\preccurlyeq\Id$ almost surely and
assume that $\E M^2$ is invertible. Then, for every $r\in(0,1]$,
\[
 \E\Tr\left[(\E M^2)^{-1}\left(M^2-\E M^2\right)^2\right]
 \leq\frac{16}{r^2}\E\Tr\left[(M-r\Id)^2\right].
\]
\end{lemma}
\begin{proof}
Write
\[
 \E M^2=\sum_{i=1}^d m_i\,u_i\otimes u_i
\]
in an orthonormal eigenbasis. Therefore,
\begin{align}
 \E\Tr\left[(\E M^2)^{-1}\left(M^2-\E M^2\right)^2\right]
 &=\sum_{i=1}^d\frac{\E|M^2u_i-\E(M^2u_i)|^2}{m_i}\notag\\
 &\leq\frac4{r^2}\sum_{m_i\geq r^2/4}\E|M^2u_i-\E(M^2u_i)|^2
 \!+\sum_{m_i<r^2/4}\frac{\E|M^2u_i-\E(M^2u_i)|^2}{m_i}.\label{eq:matrix-split}
\end{align}
For the first sum in \eqref{eq:matrix-split}, since $\E(M^2u_i)$ is the orthogonal projection of $M^2u_i$ onto the constant
vectors, for all $i$ such that $m_i\geq r^2/4$,
\begin{align*}
 \E|M^2u_i-\E(M^2u_i)|^2
 &\leq\E\left|M^2u_i-r^2u_i\right|^2
 \leq\E\,\|M+r\Id\|^2\,|Mu_i-ru_i|^2\\
 &\leq4\E\,|(M-r\Id)u_i|^2,
\end{align*}
where we used only the facts that $M$ and $r\Id$ commute, and that $\|M+r\Id\|\leq2$ since $r\in(0,1]$
and $0\preccurlyeq M\preccurlyeq\Id$. This settles the first term on the right-hand side. It remains to control the
second sum. Consider $m_i<r^2/4$. Since $0\preccurlyeq M\preccurlyeq\Id$,
\begin{equation}\label{eq:small-eigenvalue}
 \E\left|M^2u_i-\E(M^2u_i)\right|^2
 \leq\E|M^2u_i|^2=\E\langle M^4u_i,u_i\rangle
 \leq\E\langle M^2u_i,u_i\rangle=m_i,
\end{equation}
where we have used the orthogonal projection property. Furthermore, we observe that
\[
 \E|Mu_i|^2=\langle\E M^2u_i,u_i\rangle=m_i.
\]
Then,
\begin{equation}\label{eq:distance-scalar}
 \left(\E\,|(M-r\Id)u_i|^2\right)^{1/2}
 \geq r-\left(\E|Mu_i|^2\right)^{1/2}
 =r-\sqrt{m_i}>\frac r2,
\end{equation}
where we have used the triangle inequality for the first inequality. Combining \eqref{eq:small-eigenvalue} and \eqref{eq:distance-scalar}, we
obtain
\[
 \frac{\E\left|M^2u_i-\E(M^2u_i)\right|^2}{m_i}
 \leq1\leq\frac4{r^2}\E\,|(M-r\Id)u_i|^2.
\]
Thus, combining the estimates for both sums in \eqref{eq:matrix-split}, we obtain
\[
 \E\Tr\left[(\E M^2)^{-1}(M^2-\E M^2)^2\right]
 \leq\frac{16}{r^2}\sum_{i=1}^d\E\,|(M-r\Id)u_i|^2
 =\frac{16}{r^2}\E\Tr\left[(M-r\Id)^2\right],
\]
which concludes the proof.
\end{proof}
Throughout the following proof, we set $t_1$ as defined at the beginning of this section. In particular,
for all $t\leq t_1$, the concentration estimate \eqref{eq:concentration} holds.

\begin{proof}[Proof of the main theorem]
Combining Corollary \ref{cor:entropy-clt} with Lemma \ref{lem:inverse-mean} yields the following bound :
\begin{align*}
 H(\mu^n|\gamma)
 &\leq\frac1{2n}\int_0^\infty(1+C_P(\mu)t)\,
 \E\Tr\left[(A_t^2-\E A_t^2)^2(\E A_t^2)^{-1}\right]dt\\
 &\leq\frac1{2n}\int_0^\infty(1+C_P(\mu)t)^3\,
 \E\Tr\left[(A_t^2-\E A_t^2)^2\right]dt.
\end{align*}
Now, we split the integral into two parts, over $[0,t_1]$ and $[t_1,\infty)$. To begin, we analyze the integral
over $[0,t_1]$. We estimate the expectation as follows
\begin{align}
 \E\Tr\left[(A_t^2-\E A_t^2)^2\right]
 &\leq\E\Tr\left[\left(A_t-\frac{\Id}{1+t}\right)^2
 \left(A_t+\frac{\Id}{1+t}\right)^2\right]\notag\\
 &\leq\E\left[\left\|A_t+\frac{\Id}{1+t}\right\|^2
 \Tr\left(\left(A_t-\frac{\Id}{1+t}\right)^2\right)
 \ind_{\{\|A_t\|\leq2\}}\right]\notag\\
 &\quad+\E\left[\left\|A_t-\frac{\Id}{1+t}\right\|^2
 \Tr\left(\left(A_t+\frac{\Id}{1+t}\right)^2\right)
 \ind_{\{\|A_t\|>2\}}\right],\label{eq:fluctuation-events}
\end{align}
for all $t\in[0,t_1]$. For the second term, \eqref{eq:pointwise-covariance} implies
$\|A_t-\frac{\Id}{1+t}\|\leq\max\{\|A_t\|;\frac1{1+t}\}\leq\frac1t$ and thus
\begin{align}
 \E\left[\Tr\left(\left\|A_t-\frac{\Id}{1+t}\right\|^2
 \left(A_t+\frac{\Id}{1+t}\right)^2\right)\ind_{\{\|A_t\|>2\}}\right]
 &\leq\frac1{t^2}\E\left[
 \Tr\left(\left(\frac{\Id}{t}+\frac{\Id}{1+t}\right)^2\right)
 \ind_{\{\|A_t\|>2\}}\right]\notag\\
 &\leq\frac1{t^2}\Tr\left(\frac{4\Id}{t^2}\right)\Prob(\|A_t\|>2)\notag\\
 &\leq\frac{4d}{t^4}\Prob(\|A_t\|>2).\label{eq:bad-event}
\end{align}
Thus, combining \eqref{eq:bad-event} and \eqref{eq:concentration}, and using $1+C_P(\mu)t\leq1+C_P(\mu)$ for $t\leq t_1\leq1$ (up to adjusting
the universal constant $c$, we obtain :
\begin{align}
 &4\int_0^{t_1}(1+C_P(\mu)t)^3
 \E\left[\Tr\left(\left\|A_t-\frac{\Id}{1+t}\right\|^2
 \left(A_t+\frac{\Id}{1+t}\right)^2\right)
 \ind_{\{\|A_t\|>2\}}\right]dt\notag\\
 &\qquad\leq16(1+C_P(\mu))^3\int_0^{t_1}
 \frac d{t^4}\Prob(\|A_t\|>2)\,dt\notag\\
 &\qquad\leq16(1+C_P(\mu))^3\int_0^{t_1}
 \frac d{t^4}e^{-c\cdot t^{-1}}dt.\label{eq:bad-integral}
\end{align}
This quantity is bounded above by a universal constant. Indeed, taking $c'\leq c/2$ in the definition
of $t_1$, we obtain
\[
 d\int_0^{t_1}\frac{e^{-c/t}}{t^4}\,dt
 \leq de^{-c/(2t_1)}\int_0^\infty\frac{e^{-c/(2t)}}{t^4}\,dt
 =d^{1-c/(2c')}\int_0^\infty\frac{e^{-c/(2t)}}{t^4}\,dt=O(1),
\]
because $1-c/(2c')\leq0$ and the last integral is finite. Now, on the event $\{\|A_t\|\leq2\}$,
$\|A_t+\frac{\Id}{1+t}\|\leq3$.
Thus,
\begin{align}
 &4\int_0^{t_1}(1+C_P(\mu)t)^3
 \E\left[\Tr\left(\left\|A_t+\frac{\Id}{1+t}\right\|^2
 \left(A_t-\frac{\Id}{1+t}\right)^2\right)
 \ind_{\{\|A_t\|\leq2\}}\right]dt\notag\\
 &\quad\leq36(1+C_P(\mu))^3\int_0^{t_1}
 \E\Tr\left[\left(A_t-\frac{\Id}{1+t}\right)^2\right]dt\notag\\
 &\quad\leq36(1+C_P(\mu))^3\int_0^\infty(1+t)\,
 \E\Tr\left[\left(A_t-\frac{\Id}{1+t}\right)^2\right]dt\notag\\
 &\quad=72(1+C_P(\mu))^3H(\mu|\gamma)\label{eq:good-integral}
\end{align}
where we used \eqref{eq:entropy-covariance} in the last passage.

We now turn to the analysis of the integral over $[t_1,\infty)$. Applying Lemma \ref{lem:matrix-fluctuation} with $M=tA_t$ and
$r=\frac t{1+t}$ yields
\begin{equation}\label{eq:tail-fluctuation}
 \E\Tr\left[(A_t^2-\E A_t^2)^2(\E A_t^2)^{-1}\right]
 \leq16\left(\frac{1+t}{t}\right)^2
 \E\Tr\left[\left(A_t-\frac{\Id}{1+t}\right)^2\right].
\end{equation}
Then,
\begin{align}
 &\frac1{2n}\int_{t_1}^\infty(1+C_P(\mu)t)
 \E\Tr\left[(A_t^2-\E A_t^2)^2(\E A_t^2)^{-1}\right]dt\notag\\
 &\leq\frac{16}{n}\int_{t_1}^\infty\frac{(1+t)^2(1+C_P(\mu)t)}{t^2}
 \E\Tr\left[\left(A_t-\frac{\Id}{1+t}\right)^2\right]dt\notag\\
 &\leq\frac{32}{n}\frac{(1+t_1)(1+C_P(\mu)t_1)}{t_1^2}H(\mu|\gamma)\label{eq:logarithmic-tail}
\end{align}

Combining \eqref{eq:logarithmic-tail}, \eqref{eq:good-integral} and \eqref{eq:bad-integral}, we have $$H(\mu^n\lvert\gamma) \le \frac 1n\left(C_1+\left(72(1+C_P(\mu))^3 + \frac{32(1+t_1)(1+C_P(\mu)t_1)}{t_1^2}\right)H(\mu\lvert\gamma)\right).$$

where $C_1$ is a universal constant. Finally, as mentioned in the introduction, the resolution of the KLS conjecture in \cite{Biz26b} or \cite{Son26} implies the existence
of a universal constant $C_{KLS} > 0$ such that $C_P(\mu) \le C_{KLS}$ for every isotropic log-concave
probability measure $\mu$. Since the function $t\mapsto(1+t)(1+C_P(\mu)t)/t^2$ is non-increasing on $(0,\infty)$, we can bound the quotient by
$C'/t_1^2$ where $C'$ is a universal constant depending only on $C_{KLS}$. Using the definition of $t_1$ in \eqref{eq:concentration}, we finally obtain 
\[
 H(\mu^n|\gamma)\lesssim\frac1n\left(1+(\log d)^2H(\mu|\gamma)\right),
\]
which proves the first bound.

Next, we exploit the second bound provided by Lemma \ref{lem:inverse-mean}. Combining Corollary \ref{cor:entropy-clt} with Lemma
\ref{lem:inverse-mean} yields
\begin{equation}\label{eq:entropy-second-bound}
 H(\mu^n|\gamma)\leq\frac1{2n}\int_0^\infty(1+C_P(\mu)t)\,
 \E\Tr\left[(A_t^2-\E A_t^2)^2(\E A_t^2)^{-1}\right]dt.
\end{equation}

We again split the integral \eqref{eq:entropy-second-bound} into two parts. On $[0,1]$, by \eqref{eq:fourth-moment} and Lemma \ref{lem:inverse-mean}, we have
\begin{align}
 \int_0^1(1+C_P(\mu)t)
 \E\Tr\left[(A_t^2-\E A_t^2)^2(\E A_t^2)^{-1}\right]dt
 &\leq\int_0^1(1+C_P(\mu)t)^3
 \E\Tr\left[(A_t^2-\E A_t^2)^2\right]dt\notag\\
 &\leq\int_0^1(1+C_P(\mu)t)^3\E\Tr[A_t^4]dt\notag\\
 &\leq\frac{Cd}{4C_P(\mu)}(1+C_P(\mu))^4,\label{eq:small-time-second}
\end{align}
where we have used the fact that $C_P(\mu)\geq1$ for every isotropic and log-concave measure $\mu$. On
$[1,\infty)$, as for \eqref{eq:logarithmic-tail}, Lemma \ref{lem:matrix-fluctuation} and \eqref{eq:entropy-barycenter} yield
\begin{align}
 &\frac1{2n}\int_1^\infty(1+C_P(\mu)t)
 \E\Tr\left[(A_t^2-\E A_t^2)^2(\E A_t^2)^{-1}\right]dt\notag\\
 &\leq\frac8n\int_1^\infty\frac{(1+C_P(\mu)t)(1+t)^2}{t^2}
 \E\Tr\left[\left(A_t-\frac{\Id}{1+t}\right)^2\right]dt\notag\\
 &\leq\frac{32(1+C_P(\mu))}{n}H(\mu|\gamma),\label{eq:large-time-second}
\end{align}
where we have used that the function $t\mapsto\frac{(1+C_P(\mu)t)(1+t)}{t^2}$ is decreasing, then bounded by $2(1+C_P(\mu))$ on $[1,\infty)$. Combining \eqref{eq:large-time-second} and \eqref{eq:small-time-second} and using the bound $C_P(\mu)\le C_{KLS}$, we obtain 
\[
 H(\mu^n|\gamma)\leq\frac{C_2}{n}\left(d+H(\mu|\gamma)\right),
\]
where $C_2=\frac{C(1+C_{KLS})^4}{4}+32(1+C_{KLS})$. Finally, the slicing theorem \cite{Kla24c} implies that
$H(\mu|\gamma)\leq Cd$ for isotropic log-concave probability measures, this proves the second bound.
\end{proof}
\begin{remark}\label{rem:wasserstein}
Applying Talagrand's inequality to Theorem \ref{thm:main} yields
\[
 W_2(\mu^n,\gamma)\lesssim
 \sqrt{\frac1n\left(1+(\log d)^2H(\mu\mid\gamma)\right)},
\]
where $W_2(\nu_1,\nu_2)$ denotes the Wasserstein distance between the probability measures $\nu_1$ and $\nu_2$.
Moreover, the authors of \cite{Cou19} obtain
\[
 W_2(\mu^n,\gamma)\lesssim\sqrt{\frac{dC_P(\mu)}{n}}.
\]
Hence, our estimate does not outperform that of \cite{Cou19} in general, even if the KLS conjecture
is true, but it has the advantage of incorporating the relative entropy $H(\mu\mid\gamma)$, which can provide
finer control when $\mu$ is close to the Gaussian distribution. Applying Talagrand's inequality to our
second bound yields a Wasserstein estimate of the same order as theirs.
\end{remark}
We adopt the same strategy to prove Theorem \ref{thm:uniform} :
\begin{proof}[Proof of Theorem \ref{thm:uniform}]
Since $\mu_t$ is $(\kappa+t)$-uniformly log-concave, we have
\begin{equation}\label{eq:uniform-covariance}
 0\preccurlyeq A_t\preccurlyeq\frac1{\kappa+t}\Id,\qquad t\geq0.
\end{equation}
Consequently, $\E A_t^2\preccurlyeq(\kappa+t)^{-1}\E A_t$. Together with \eqref{eq:mean-derivative}, this gives
\[
 \frac{d}{dt}((\kappa+t)\E A_t)=\E A_t-(\kappa+t)\E A_t^2\succcurlyeq0.
\]
Since $A_0=\Id$, it follows that $\E A_t\succcurlyeq\frac{\kappa}{\kappa+t}\Id$ and hence
\begin{equation}\label{eq:uniform-inverse}
 \|(\E A_t)^{-1}\|\leq\frac{\kappa+t}{\kappa}.
\end{equation}
We apply Lemma \ref{lem:matrix-fluctuation} with $M=(\kappa+t)A_t$ and $r=(\kappa+t)/(1+t)$. By \eqref{eq:uniform-covariance} and $\kappa\leq1$, we have
$0\preccurlyeq M\preccurlyeq\Id$ and $0<r\leq1$. Hence,
\[
 \E\Tr\left[(A_t^2-\E A_t^2)^2(\E A_t^2)^{-1}\right]
 \leq16\left(\frac{1+t}{\kappa+t}\right)^2
 \E\Tr\left[\left(A_t-\frac{\Id}{1+t}\right)^2\right].
\]
Combining this estimate with Corollary \ref{cor:entropy-clt} and \eqref{eq:uniform-inverse}, we obtain
\begin{align*}
 H(\mu^n|\gamma)&\leq\frac8{n\kappa}\int_0^\infty\frac{(1+t)^2}{\kappa+t}
 \E\Tr\left[\left(A_t-\frac{\Id}{1+t}\right)^2\right]dt\\
 &\leq\frac8{n\kappa^2}\int_0^\infty(1+t)\,
 \E\Tr\left[\left(A_t-\frac{\Id}{1+t}\right)^2\right]dt\\
 &\leq\frac{16}{n\kappa^2}H(\mu|\gamma),
\end{align*}
where we used $\kappa(1+t)\leq\kappa+t$ in the second inequality and Proposition \ref{prop:entropy-covariance} in the last one.
\end{proof}


\begin{thebibliography}{99}
\bibitem[Art04]{Art04}
\textit{S. Artstein, K. Ball, F. Barthe and A. Naor}, Solution of Shannon's problem on the
monotonicity of entropy, J. Amer. Math. Soc. 17, No. 4, 975--982 (2004).

\bibitem[Bal02]{Bal02}
\textit{K. Ball, F. Barthe and A. Naor}, Entropy jumps in the presence of a spectral gap, Duke
Math. J. 119, No. 1, 41--63 (2003).

\bibitem[Bal12]{Bal12}
\textit{K. Ball and V. H. Nguyen}, Entropy jumps for isotropic log-concave random vectors and
spectral gap, Studia Math. 213, No. 1, 81--96 (2012).

\bibitem[Bar86]{Bar86}
\textit{A. R. Barron}, Entropy and the central limit theorem, Ann. Probab. 14, No. 1, 336--342
(1986).

\bibitem[Biz26]{Biz26}
\textit{P. Bizeul}, Optimal $MM^*$ bounds for convex bodies, arXiv:2607.29458 (2026).

\bibitem[Biz26b]{Biz26b}
\textit{P. Bizeul, B. Klartag and J. Lehec}, Presenting a proof of the Kannan--Lovász--Simonovits
conjecture, arXiv:2610.05474 (2026).

\bibitem[Car91]{Car91}
\textit{E. A. Carlen and A. Soffer}, Entropy production by block variable summation and central
limit theorems, Comm. Math. Phys. 140, No. 2, 339--371 (1991).

\bibitem[Che26]{Che26}
\textit{Y. Chen and B. Klartag}, Digesting the proof of the sharp thin-shell inequality,
arXiv:2607.23307 (2026).

\bibitem[Cou19]{Cou19}
\textit{T. A. Courtade, M. Fathi and A. Pananjady}, Existence of Stein kernels under a spectral
gap, and discrepancy bounds, Ann. Inst. Henri Poincaré Probab. Stat. 55, No. 2, 777--790
(2019).

\bibitem[Den26]{Den26}
\textit{C.-S. Deng, L. Wang and L. Xu}, Quantitative bounds for high dimensional entropic CLT,
arXiv:2604.05861 (2026).

\bibitem[Eld20]{Eld20}
\textit{R. Eldan, D. Mikulincer and A. Zhai}, The CLT in high dimensions: quantitative bounds
via martingale embedding, Ann. Probab. 48, No. 5, 2494--2524 (2020).

\bibitem[Gua25]{Gua25}
\textit{Q. Guan}, On tail probability of the covariance matrix in Eldan's stochastic localization,
arXiv:2508.14943 (2025).

\bibitem[Joh04]{Joh04}
\textit{O. Johnson and A. Barron}, Fisher information inequalities and the central limit theorem,
Probab. Theory Related Fields 129, No. 3, 391--409 (2004).

\bibitem[Kan95]{Kan95}
\textit{R. Kannan, L. Lovász and M. Simonovits}, Isoperimetric problems for convex bodies and
a localization lemma, Discrete Comput. Geom. 13, No. 3--4, 541--559 (1995).

\bibitem[Kla21]{Kla21}
\textit{B. Klartag and E. Putterman}, Spectral monotonicity under Gaussian convolution, Ann.
Fac. Sci. Toulouse Math. (6) 32, No. 5, 939--967 (2023).

\bibitem[Kla23]{Kla23}
\textit{B. Klartag}, Logarithmic bounds for isoperimetry and slices of convex sets, Ars Inven.
Anal., Paper No. 4, 17 pp. (2023).

\bibitem[Kla24]{Kla24}
\textit{B. Klartag and J. Lehec}, Isoperimetric inequalities in high-dimensional convex sets, Bull.
Amer. Math. Soc. 62, No. 4, 575--642 (2025).

\bibitem[Kla24b]{Kla24b}
\textit{B. Klartag and O. Ordentlich}, The strong data processing inequality under the heat flow,
IEEE Trans. Inform. Theory 71, No. 5, 3317--3333 (2025).

\bibitem[Kla24c]{Kla24c}
\textit{B. Klartag and J. Lehec}, Affirmative resolution of Bourgain's slicing problem using
Guan's bound, Geom. Funct. Anal. 35, 1147--1168 (2025).

\bibitem[Kla25]{Kla25}
\textit{B. Klartag and J. Lehec}, Thin-shell bounds via parallel coupling, arXiv:2507.15495 (2025).

\bibitem[Leh13]{Leh13}
\textit{J. Lehec}, Representation formula for the entropy and functional inequalities, Ann. Inst.
Henri Poincaré Probab. Stat. 49, No. 3, 885--899 (2013).

\bibitem[Lip01]{Lip01}
\textit{R. S. Liptser and A. N. Shiryaev}, Statistics of random processes I: General theory, 2nd
ed., Stochastic Modelling and Applied Probability, Vol. 5, Springer, Berlin (2001).

\bibitem[Son26]{Son26}
\textit{Z. Song and X. Zhang}, An $O(1)$ bound for the KLS constant, arXiv:2610.01447 (2026).
\end{thebibliography}
\end{document}